\documentclass[titlepage,12pt]{article}

\pdfoutput=1

\newif\ifmtr
\mtrtrue 

\newif\ifreleased
\releasedtrue 

\ifreleased
\else
\fi

\usepackage{makeidx}
\usepackage{graphics}
\usepackage[matrix,arrow,curve]{xy}
\usepackage{amssymb}
\usepackage{amsmath}
\usepackage{amsthm}
\usepackage{url}
\usepackage{hyperref}\hypersetup{colorlinks=true,}

\newtheorem{thm}{Theorem}
\newtheorem{lem}{Lemma}
\newcommand{\cpsa}{\textsc{cpsa}}
\newcommand{\pvs}{\textsc{pvs}}
\newcommand{\acp}{\textsc{acp}}
\newcommand{\cn}[1]{\ensuremath{\operatorname{\mathsf{#1}}}}

\newcommand{\fn}[1]{\ensuremath{\operatorname{\mathit{#1}}}}
\newcommand{\srt}[1]{\ensuremath{\mathsf{#1}}}
\newcommand{\typ}{\mathbin:}
\newcommand{\sdom}{\fn{Dom}}

\newcommand{\seq}[1]{\ensuremath{\langle#1\rangle}}
\newcommand{\enc}[2]{\ensuremath{\{\!|#1|\!\}_{#2}}}
\newcommand{\invk}[1]{{#1}^{-1}}
\newcommand{\inbnd}{\mathord -}
\newcommand{\outbnd}{\mathord +}
\newcommand{\nat}{\ensuremath{\mathbb{N}}}

\newcommand{\all}[1]{\forall#1\mathpunct.}
\newcommand{\some}[1]{\exists#1\mathpunct.}
\newcommand{\funct}[1]{\lambda#1\mathpunct.}
\newcommand{\pow}[1]{\mathcal P(#1)}
\newcommand{\prefix}[2]{#1\mid#2}
\newcommand{\init}{\fn{init}}
\newcommand{\resp}{\fn{resp}}
\newcommand{\run}{\mathcal{R}}
\newcommand{\pt}{\fn{pt}}
\newcommand{\form}{\mathcal{K}}
\newcommand{\sent}{\mathcal{S}}

\newcommand{\interp}{\mathcal{I}}
\newcommand{\alg}[1]{\ensuremath{\mathfrak#1}}
\newcommand{\tr}{\ensuremath{\mathfrak C}}
\newcommand{\rl}{\fn{rl}}
\newcommand{\skel}{\mathsf{k}}
\newcommand{\insta}{\mathsf{i}}

\newcommand{\evt}{\fn{evt}}
\newcommand{\role}{\mathsf{r}}
\newcommand{\orig}{\mathcal{O}}

\newcommand{\tran}{\ensuremath{\tau}}
\newcommand{\pth}{\ensuremath{\pi}}
\newcommand{\type}{\ensuremath{\mathfrak T}}
\newcommand{\up}{\mathord\uparrow}
\newcommand{\down}{\mathord\downarrow}

\makeindex

\title{Proving Security Goals With \\ Shape Analysis Sentences}
\author{John D.\ Ramsdell}

\begin{document}
\ifmtr




\begin{titlepage}
\begin{trivlist}\sffamily\bfseries\large
\item
MTR130488\\[-1.2ex]
\hrule ~\\
{\mdseries MITRE TECHNICAL REPORT}\\[1cm]
\LARGE
Proving Security Goals With\\[1ex] Shape Analysis Sentences\\[2.5cm]
\large
September 2013\\
~\\
\mdseries
John D.~Ramsdell
\vfill
\normalsize
\bfseries
\begingroup\footnotesize
\begin{tabbing}
Sponsor: \phantom{spo} \= NSA/R2D \phantom{phantom} \=
Contract No.: \phantom{pro}\= W15P7T-13-C-F600 \\
Dept. No.: \>G063 \>Project No.: \>0713N6BZ-TD\\[1.2cm]
The views, opinions and/or findings contained in \>\>\>\phantom{space}
\ifreleased
\= Approved for Public Release\\
this report are those of The MITRE Corporation
\>\>\>\phantom{space} Case No. 13-3482\\
and should not be construed as an official\\
\fi
Government position, policy, or decision, unless\\
designated by other documentation.\\[4mm]
{\copyright} 2013 The MITRE Corporation. All Rights Reserved.
\end{tabbing}
\endgroup
~\\
\noindent
\includegraphics{mitrelogo-0.mps}\\
Center for Integrated Intelligence Systems\\
Bedford, Massachusetts
\end{trivlist}
\end{titlepage}

\noindent
Approved by:\\[1in]
\rule{3in}{.3mm}\\
Paul D. Rowe, 0713N6BZ-TD Project Leader

\clearpage

\else
\maketitle
\fi

\begin{abstract}
The paper that introduced shape analysis sentences presented a method
for extracting a sentence in first-order logic that completely
characterizes a run of {\cpsa}.  Logical deduction can then be used to
determine if a security goal is satisfied.

This paper presents a method for importing shape analysis sentences
into a proof assistant on top of a detailed theory of strand spaces.
The result is a semantically rich environment in which the validity of
a security goal can be determined using shape analysis sentences and the
foundation on which they are based.
\end{abstract}

\tableofcontents
\listoffigures

\clearpage

\section{Introduction}

A central problem in cryptographic protocol analysis is to determine
whether a formula that expresses a security goal about behaviors
compatible with a protocol is true.  Following~\cite{Guttman09}, a
security goal is a quantified implication:
\begin{equation}\label{eqn:security goal}
\all{\vec{x}}\Phi_0\supset\bigvee_{1\le i\le n}\some{\vec{y}_i}\Phi_i.
\end{equation}

The hypothesis~$\Phi_0$ is a conjunction of atomic formulas describing
regular (honest) behavior.  Each disjunct~$\Phi_i$ that makes up the
conclusion is also a conjunction of atomic formulas.  When~$\Phi_i$
describes desired behaviors of other regular participants, then the
formula is an \index{authentication goal}\emph{authentication} goal.
The goal says that each run of the protocol compatible with~$\Phi_0$
will include the regular behavior described by one of the disjuncts.
When $n=0$, the goal's conclusion is false.  In this case, if~$\Phi_0$
mentions an unwanted disclosure, Eq.~\ref{eqn:security goal} says the
disclosure cannot occur, thus a security goal with $n=0$ expresses a
\index{secrecy goal}\emph{secrecy} goal.

Guttman~\cite{Guttman09} presented a model-theoretic approach to
establishing security goals in the context of strand space theory.  In
that setting, a \index{skeleton}\emph{skeleton} describes regular
behaviors compatible with a protocol.  For skeleton~$k$ and
formula~$\Phi$, he defined $k,\alpha\models\Phi$ to mean that the
conjunction of atomic formulas that make up~$\Phi$ is satisfied in~$k$
with variable assignment~$\alpha$.

A \index{realized skeleton}\emph{realized} skeleton is one that
includes enough regular behavior to specify all the non-adversarial
part of an execution of the protocol.  In a realized skeleton, its
message transmissions combined with possible adversarial behavior
explain every message reception in the skeleton.

In strand space theory, a \index{homomorphism}\emph{homomorphism} is a
structure-preserving map~$\delta$ that shows how the behaviors in one
skeleton are reflected within another.  As skeletons serve as models,
homomorphisms preserve satisfaction for conjunctions of atomic
formulas.

The Cryptographic Protocol Shapes Analyzer ({\cpsa}) constructs
homomorphisms from a skeleton~$k_0$ to realized
skeletons~\cite{cpsa09}.  If {\cpsa} terminates, it generates a set of
realized skeletons~$k_i$ and a set of homomorphisms $\delta_i\typ
k_0\mapsto k_i$.  These realized skeletons are all the minimal,
essentially different skeletons that are homomorphic images of~$k_0$
and are called the \index{shapes}\emph{shapes} of the analysis.

Ramsdell~\cite{Ramsdell12} described {\cpsa}'s support for security
goals. {\cpsa} includes a tool that extracts a sentence that
characterizes a shape analysis.  This so called \index{shape analysis
  sentence}\emph{shape analysis sentence} is special in that it
encodes everything that can be learned from the shape analysis.

Given a shape analysis sentence, a security goal is achieved if the
goal can be deduced from the sentence.  {\cpsa} includes a Prolog
program that translates shape analysis sentences into
\index{Prover9}Prover9~\cite{prover9} syntax.  Typically, a goal that
is a theorem is quickly proved by Prover9.

There is another advantage to this approach.  It can be tedious to
generate security goals.  Realistic ones can be large and complicated.
An easy way to create one is to modify a shape analysis sentence.
This typically involves deleting parts of the conclusion.

There is a disadvantage to this approach.  When a goal cannot be
deduced from a shape analysis sentence, one cannot conclude that there
is a counterexample.  It could be simply that the sentence is not
relevant to the security goal.  It could also be that a proof of the
goal depends on a fact not exposed by a shape analysis sentence.  For
example, the precedes relation on nodes in a skeleton is transitive,
but that fact is not available to Prover9.

This paper describes the method that was used to import shape analysis
sentences into the proof assistant
\index{PVS@\pvs}{\pvs}~\cite{cade92-pvs} on top of a detailed theory
of strand spaces specified in~{\pvs}.  In this environment, if the
proof of a security goal depends on the transitivity of the precedes
relation, that fact is available as a lemma.  Furthermore, if a
security goal is false, one can construct a counterexample and use it
to prove the security goal is in fact false.

In {\cpsa}, executions of protocols are represented by skeletons.
Associated with each skeleton is a free message algebra generated by a
finite set of variables.  Skeletons are used as models in the original
paper on shape analysis sentences.

The {\pvs} strand space theory uses bundles over an initial algebra as
its representation of executions of protocols.  This allows for a
shallow embedding of strand space theory in which algebra variables
are replaced by logical variables in {\pvs}.  This specification
choice alleviates the need to manipulate homomorphisms within {\pvs}.
Section~\ref{sec:sas} contains two descriptions that relate skeletons
to bundles.

\paragraph{Motivating Example.}

The running example used throughout this paper is now presented.  An
informal version of the example is presented here, and the example
with all of the details filled in is in Section~\ref{sec:example}.

The following simple example protocol is due to Bruno
Blanchet~\cite{BlanchetHDR}.
$$\begin{array}{r@{{}:{}}l}
A\to B&\enc{\enc{s}{\invk{a}}}{b}\\
B\to A&\enc{d}{s}
\end{array}$$
Alice~($A$) freshly generates symmetric key~$s$, signs the symmetric
key with her private uncompromised asymmetric key $\invk{a}$ and intends
to encrypt it with Bob's~($B$) uncompromised asymmetric key $b$.
Alice expects to receive data~$d$ encrypted, such that only Alice and
Bob have access to it.

The protocol was constructed with a known flaw for expository
purposes, and as a result the secret is exposed due to an
authentication failure.  The protocol does not prevent Alice from
using a compromised key~$b'$, so that Mallory~($M$) and Eve~($E$) can
perform this man-in-the-middle attack:
$$\begin{array}{r@{{}:{}}l}
A\to M&\enc{\enc{s}{\invk{a}}}{b'}\\
M\to B&\enc{\enc{s}{\invk{a}}}{b}\\
B\to E&\enc{d}{s}
\end{array}$$

The protocol fails to provide a means for Bob to ensure the original
message was encrypted using his key.  The authentication failure is
avoided with this variation of the protocol:
\begin{equation}\label{eq:amended protocol}
\begin{array}{r@{{}:{}}l}
A\to B&\enc{\enc{s,b}{\invk{a}}}{b}\\
B\to A&\enc{d}{s}
\end{array}
\end{equation}

In strand space theory, a \emph{strand}\index{strand} is a linearly
ordered sequence of events $e_0\Rightarrow\cdots\Rightarrow e_{n-1}$,
and an \index{event}\emph{event} is either a message transmission
$\bullet\to$ or a reception $\bullet\gets$.  In {\cpsa}, adversarial
behavior is not explicitly represented, so strands always represent
regular behavior.

Regular behavior is constrained by a set of roles that make up the
protocol.  In this protocol, Alice's behaviors must be compatible with
an initiator role, and Bob's behaviors follow a responder role.
\begin{equation}\label{eq:protocol}
\begin{array}{r@{\qquad}l}
\xymatrix@=.6em{\raisebox{-1ex}[0ex][0ex]{\strut\init}\\
  \bullet\ar@{=>}[d]\ar[r]&\enc{\enc{s}{\invk{a}}}{b}\\
  \bullet&\enc{d}{s}\ar[l]}&
\xymatrix@=.6em{&\raisebox{-1ex}[0ex][0ex]{\strut\resp}\\
\enc{\enc{s}{\invk{a}}}{b}\ar[r]&\bullet\ar@{=>}[d]\\
\enc{d}{s}&\bullet\ar[l]}
\end{array}
\end{equation}

The important authentication goal from Bob's perspective is that if an
instance of a responder role runs to completion, there must have been
an instance of the initiator role that transmitted its first message.
Furthermore, assuming the symmetric key is freshly generated, and the
private keys are uncompromised, the two strands agree on keys used for
signing and encryption.

A {\cpsa} analysis of the authentication goal requires two inputs, a
specification of the roles that make up the protocol, as in
Eq.~\ref{eq:protocol}, and a question about runs of the protocol.  The
question in this case is the hypothesis of Eq.~\ref{eq:shape
  analysis}, that an instance of the responder role ran to completion.
In these diagrams, a strand instantiated from a role is distinguished
from a role by placing messages above communication arrows, and
$\succ$ is used to assert an event occurred after another.

\begin{equation}\label{eq:shape analysis}
\begin{array}{rcl}
\xymatrix@R=1em@C=3em{\raisebox{-0.5ex}[0ex][0ex]{\strut\resp}\\
\bullet\ar@{=>}[d]&\ar[l]_{\enc{\enc{s}{\invk{a}}}{b}}\\
\bullet\ar[r]^{\enc{d}{s}}&}&
\raisebox{-4.0ex}{implies}&
\xymatrix@R=1em@C=3em{\raisebox{-0.5ex}[0ex][0ex]{\strut\resp}&&
\raisebox{-0.5ex}[0ex][0ex]{\strut\init}\\
\bullet\ar@{=>}[d]&\succ\ar[l]_{\enc{\enc{s}{\invk{a}}}{b}}&\bullet
\ar[l]_{\enc{\enc{s}{\invk{a}}}{b'}}\\
\bullet\ar[r]^{\enc{d}{s}}&}
\end{array}
\end{equation}

{\cpsa} produces the conclusion in Eq.~\ref{eq:shape analysis}, that
an instance of the initiator role must have transmitted its first
message, but it does not conclude that the strands agree on the key
used for the outer encryption.  When {\cpsa} is run using the amended
protocol in Eq.~\ref{eq:amended protocol}, the strands agree on the
key, and the authentication goal is achieved.

The contribution of this paper is a method of importing security goals
and the results of a {\cpsa} analysis into {\pvs} such that proofs
about the goals can rely on a detailed theory of strand spaces.  The
shape analysis sentence associated with this example is presented in
Section~\ref{sec:example}.

\paragraph{Some Related Work.}

This paper is the result of implementing security goals as described
by Guttman in~\cite{Guttman09}.  The original motivation for extracting
shape analysis sentences rather than following the procedure
in~\cite{Guttman09} was ease of implementation.  With shape analysis
sentences, most of the work is performed by a post-processing stage,
and there were only a few changes made to the core {\cpsa} program.
Only later was it realized the sense in which shape analysis sentences
completely characterize a shape analysis.

The Scyther tool~\cite{cremers06} integrates security goal
verification with its core protocol analysis algorithm.  Security
goals are easy to state as long as they can be expressed using a
predefined vocabulary, however, there is no sense in which Scyther
goals characterize an analysis.

The Protocol Composition Logic~\cite{datta05} provides a contrasting
approach to specifying security goals.  It extends strand spaces by
adding an operational semantics as a small set of reduction rules, and
a run of a protocol is a sequence of reduction steps derived from an
initial configuration.  The logic is a temporal logic interpreted over
runs.

\paragraph{Structure of this Paper.}

Section~\ref{sec:strand spaces} describes strand spaces as formalized
in {\pvs}, Section~\ref{sec:sas} reintroduces shape analysis
sentences, and Section~\ref{sec:example} displays the example above in
full detail.  Appendix~\ref{sec:role annotations} describes an
extension that can be used to prove security goals that involve
long-term state.

\paragraph{Notation.}

A finite sequence is a function from an initial segment of the natural
numbers.  The length of a sequence~$X$ is~$|X|$, and
sequence~$X=\seq{X(0),\ldots, X(n-1)}$ for $n=|X|$.  If~$S$ is a set,
then~$S^\ast$ is the set of finite sequences over~$S$, and~$S^+$ is the
non-empty finite sequences over~$S$.  The prefix of sequence~$X$ of
length~$n$ is~$\prefix{X}{n}$.

\section{Strand Spaces}\label{sec:strand spaces}

\index{PVS@\pvs}{\pvs} is based on classical, typed higher-order
logic.  It has dependent types and parameterized theories.

This section describes the {\pvs} definition of strand
spaces~\cite{ThayerEtal99} in a style motivated by the {\pvs}
language~\cite{cade92-pvs}, that is, the presentation attempts to
minimize the gap between the actual proofs and this content.

\begin{figure}
$$\begin{array}{ll@{{}\typ{}}ll}
\mbox{Sorts:}&
\multicolumn{3}{l}{\mbox{$\top$, $\srt{A}$, $\srt{S}$, $\srt{D}$}}\\
\mbox{Subsorts:}&
\multicolumn{3}{l}{\mbox{$\srt{A}<\top$, $\srt{S}<\top$, $\srt{D}<\top$}}\\
\mbox{Operations:}&(\cdot,\cdot)&\top\times\top\to\top& \mbox{Pairing}\\
&\enc{\cdot}{(\cdot)}&\top\times\srt{A}\to\top&\mbox{Asymmetric encryption}\\
&\enc{\cdot}{(\cdot)}&\top\times\srt{S}\to\top&\mbox{Symmetric encryption}\\
&\invk{(\cdot)}&\srt{A}\to\srt{A}& \mbox{Asymmetric key inverse}\\
&\invk{(\cdot)}&\srt{S}\to\srt{S}& \mbox{Symmetric key inverse}\\
&\cn{a}_i,\cn{b}_i&\srt{A}& \mbox{Asymmetric key constants}\\
&\cn{s}_i&\srt{S}& \mbox{Symmetric key constants}\\
\mbox{Equations:}&\multicolumn{2}{l}{\invk{\cn{a}_i}=\cn{b}_i\quad
\invk{\cn{b}_i}=\cn{a}_i}
&(i\in\nat)\\
&\multicolumn{2}{l}{\all{k\typ\srt{A}}\invk{(\invk{k})}=k}
&\all{k\typ\srt{S}}\invk{k}=k
\end{array}$$
\caption{Simple Crypto Algebra Signature}\label{fig:signature}
\end{figure}

\paragraph{Message Algebra.}

An order-sorted algebra~\cite{GoguenMeseguer92} is a generalization of
a many-sorted algebra in which sorts may be partially ordered.  The
carrier sets associated with ordered sorts are related by the subset
relation.

Figure~\ref{fig:signature} shows the simplification of the {\cpsa}
message algebra signature used by the examples in this paper.
Sort~$\top$ is the sort of all messages.  Messages of sort~$\srt{A}$
(asymmetric keys), sort~$\srt{S}$ (symmetric keys), and sort~$\srt{D}$
(data) are called \index{atoms}\emph{atoms}.  Messages are atoms or
constructed using encryption $\enc{\cdot}{(\cdot)}$ and pairing
$(\cdot,\cdot)$, where the comma operation is right associative and
parentheses are omitted when the context permits.

The message algebra \alg{A} is the \index{algebra!initial}initial
quotient term algebra over the signature.  The canonical
representative for each message is the term that contains no
occurrences of the inverse operation~$\invk{(\cdot)}$.  The set of
messages associated with a sort is called its
\index{carrier set}\emph{carrier set}.  The set of message algebra atoms
is~\alg{B}.

A message~$t_0$ is \index{carried by}\emph{carried by}~$t_1$, written
$t_0\sqsubseteq t_1$ if~$t_0$ can be extracted from a reception
of~$t_1$, assuming plaintext is extractable from encryptions.  In
other words,~$\sqsubseteq$ is the smallest reflexive, transitive
relation such that $t_0\sqsubseteq t_0$, $t_0\sqsubseteq (t_0, t_1)$,
$t_1\sqsubseteq (t_0, t_1)$, and $t_0\sqsubseteq\enc{t_0}{t_1}$.

\paragraph{Strand Spaces.}
A run of a protocol is viewed as an exchange of messages by a finite
set of local sessions of the protocol.  Each local session is called a
\index{strand}\emph{strand}.  The behavior of a strand, its
\index{trace}\emph{trace}, is a finite non-empty sequence of messaging
events.  An \index{event}\emph{event} is either a message transmission
or a reception.  Outbound message $t\in\alg{A}$ is written as~$\outbnd
t$, and inbound message~$t$ is written as~$\inbnd t$.  The set of
traces over $\alg{A}$ is $\tr=(\pm\alg{A})^+$.  A message
\index{origination}\emph{originates} in trace~$C$ at index~$i$ if it is
carried by $C(i)$, $C(i)$ is outbound, and it is not carried by any
event earlier in the trace.

A \index{strand space}\emph{strand space}~$\Theta$ over
algebra~$\alg{A}$ is a finite non-empty sequence of traces in $\tr$.
A strand~$s$ is a member of the domain of $\Theta$, and its trace is
$\Theta(s)$.  An atom~$t$ is
\index{non-origination}\emph{non-originating} in a strand
space~$\Theta$, written \index{non@\fn{non}}$\fn{non}(\Theta,t)$, if it
originates on no strand.

Message events occur at nodes in a strand space.  For each strand~$s$,
there is a node for every event in~$\Theta(s)$.  The
\index{nodes}\emph{nodes} of strand space $\Theta$ are $\{(s,i)\mid
s\in\sdom(\Theta), 0\leq i < |\Theta(s)|\}$, and the event at a node
is $\evt_\Theta(s,i)=\Theta(s)(i)$.  A node names an event in a strand
space.  The relation~$\Rightarrow$ defined by
$\{(s,i-1)\Rightarrow(s,i)\mid s\in\sdom(\Theta), 1\leq
i<|\Theta(s)|\}$ is called the \index{strand succession}\emph{strand
  succession relation}.  An atom~$t$ \index{unique
  origination}\emph{uniquely originates} in a strand space~$\Theta$ at
node~$n$, written \index{uniq@\fn{uniq}}$\fn{uniq}(\Theta,t,n)$, if it
originates in the trace of exactly one strand~$s$ at index~$i$, and
$n=(s,i)$.

\paragraph{Bundles.}
The pair $\Upsilon=(\Theta,\to)$ is a \index{bundle}\emph{bundle} if
it defines a directed acyclic graph, where the vertices are the nodes
of $\Theta$, and an edge represents communication~($\rightarrow$) or
strand succession~($\Rightarrow$) in~$\Theta$.  For communication, if
$n_0\rightarrow n_1$, then there is a message~$t$ such
that~$\evt_\Theta(n_0)=\outbnd t$ and~$\evt_\Theta(n_1)=\inbnd t$.
For each reception node~$n_1$, there is a unique transmission
node~$n_0$ with $n_0\rightarrow n_1$.

Each acyclic graph has a transitive irreflexive
relation~$\prec$\index{precedes@$\prec$~(precedes)} on its vertices.
The relation specifies the causal ordering of nodes in a bundle.  A
transitive irreflexive binary relation is also called a strict order.

\paragraph{Runs of Protocols.}
In a run of a protocol, the behavior of each strand is constrained by
a role in a protocol.  Adversarial strands are constrained by roles as
are non-adversarial strands.  A \index{role}\emph{role} is a set of
\emph{role items} of the form $\role(C,N,U)$, where $C\in\tr$,
$N\in\pow{\alg{B}}^+$, $U\in\pow{\alg{B}}^+$, and the lengths of $C$,
$N$, and $U$ agree.  The trace of the role item is~$C$, its
non-origination assumptions are~$N$, and its unique origination
assumptions are~$U$.  A strand is an instance of a role item in a
strand space, written $\fn{inst}(\Theta,s,\role(C,N,U))$, if for
$h=|\Theta(s)|$,
\begin{enumerate}
\item $h\leq|C|$,
\item $\prefix{C}{h}=\Theta(s)$,
\item for all $i<h$, $t\in N(i)$ implies $\fn{non}(\Theta,t)$, and
\item for all $i<h$, $t\in U(i)$ implies $\fn{uniq}(\Theta,t,(s,i))$.
\end{enumerate}

A \index{protocol}\emph{protocol} is a set of roles.  A
bundle~$\Upsilon=(\Theta,\to)$ is a \index{run of protocol}\emph{run of
  protocol} $P$ if there is a role assignment
$\rl\typ\sdom(\Theta)\to P$ such that for each $s\in\sdom(\Theta)$,
there exists $\role(C,N,U)\in\rl(s)$ such that
\index{inst@\fn{inst}}$\fn{inst}(\Theta,s,\role(C,N,U))$.  Let~$\run_P$ be
the set of bundles that are runs of protocol~$P$.

The description of roles differs from most presentations.  Role
origination assumptions usually are specified by a set of atoms,
instead of a sequence of sets of atoms.  The {\pvs} theory follows the
technique used in the {\cpsa} implementation.  A sequence is used so
as to make explicit the length of the instance of a role at which each
origination assumption applies.  Furthermore, roles are normally
described as templates to be copied and refined, rather than as sets
of role items.  This difference will be addressed in the next section.

\begin{figure}
$$\begin{array}{r@{{}={}}l}
\fn{create}(t\in\alg{B})&\seq{\outbnd t}\\
\fn{pair}(t_0\typ\top, t_1\typ\top)&
\seq{\inbnd t_0,\inbnd t_1,\outbnd (t_0,t_1)}\\
\fn{sep}(t_0\typ\top, t_1\typ\top)&
\seq{\inbnd (t_0, t_1),\outbnd t_0,\outbnd t_1}\\
\fn{enc}(t\typ\top, k\typ\srt{A}|\srt{S})&
\seq{\inbnd t,\inbnd k,\outbnd \enc{t}{k}}\\
\fn{dec}(t\typ\top, k\typ\srt{A}|\srt{S})&
\seq{\inbnd \enc{t}{k},\inbnd\invk{k},\outbnd t}
\end{array}$$
\caption{Adversary Traces}\label{fig:adversary}
\end{figure}

\paragraph{Adversary Model.}
The traces of the roles that constrain adversarial
behavior\index{adversary} are in Figure~\ref{fig:adversary}.  For the
encryption related traces, $k\typ\srt{A}|\srt{S}$ asserts that
$k\typ\srt{A}$ or $k\typ\srt{S}$.  There are no origination
assumptions in the adversary's roles.

The parameter of the \fn{create} role is restricted to atoms.  In
fact, the defining characteristic of an atom is it denotes the set of
messages the adversary can create out of thin air modulo origination
assumptions.

\section{Importing Protocol Analyses}\label{sec:sas}

Unlike the {\pvs} theories, {\cpsa} does not use bundles as its
representation of runs of a protocol.  Instead, it uses abstract
interpretation to discuss sets of bundles using an object called a
skeleton.

\paragraph{Skeletons.}

Skeletons and bundles share the same signature, but their algebras
differ.  Rather than using the initial algebra, each skeleton has a
\index{algebra!free}free algebra generated from a finite set of
variables.  Subscripting is used to indicate when a free algebra is in
use.  Thus, if~$X$ is a set of variables along with their sorts, then
$\Theta_X$ is a strand space over the free algebra generated
by~$X$,~$\alg{A}_X$.

The treatment of roles is slightly different in {\cpsa}.  The {\pvs}
theories define a role as a set of role items as described earlier.
In {\cpsa}, a role is a template that is instantiated to produce the
equivalent of a role item via an algebra homomorphism~$\sigma$.  Thus
for {\cpsa} role $r=\role(C_X,N_X,U_X)$, the related role item-like
object is $\role(\sigma\circ C_X,\sigma\circ N_X,\sigma\circ U_X)$,
which by abuse of notation, we write as $\sigma(r)$.  A {\pvs} role is
\index{template inspired role}\emph{template inspired} by
$r=\role(C_X,N_X,U_X)$ if it is of the form $\{\sigma(r)\mid
\sigma\in\alg{A}_X\to\alg{A}\}$.

Associated with each skeleton is protocol~$P$ as a set of roles in
template form, and a strand space~$\Theta_X$.  In {\cpsa} syntax, the
trace and role associated with a strand is specified by an
\index{instance}\emph{instance}.  An instance is of the form
$\insta(r,h,\sigma)$, where~$r\in P$ is a role, $h$ specifies the
length of a trace instantiated from the role, and~$\sigma$ specifies
how to instantiate the variables in the role to obtain the trace.
Thus the trace in~$\tr_X$ associated with
$\insta(\role(C_Y,U_Y,N_Y),h,\sigma)$ is $\sigma\circ\prefix{C_Y}{h}$,
the prefix of length~$h$ that results from applying~$\sigma$ to~$C_Y$,
where~$\sigma$ is a homomorphism from~$\alg{A}_Y$ to~$\alg{A}_X$.

A \index{skeleton}\emph{skeleton} has the form
$\skel(P,I_X,\prec,N_X,U_X)$, where~$P$ is the protocol, $I_X$ is an
instance map, \index{precedes@$\prec$~(precedes)}$\prec$ is a strict
node ordering, $N_X$ is a set of atoms assumed to be non-originating,
and~$U_X$ is a set of atoms assumed to be uniquely originating.  The
instance map~$I_X$ is a finite non-empty sequence of instances, where
the range of the homomorphism associated with each instance
is~$\alg{A}_X$.

The strand space associated with a skeleton is defined by its instance
map.  When $I_X(s)=\insta(\role(C_Y,U_Y,N_Y),h,\sigma)$, trace
$\Theta_X(s)=\sigma\circ\prefix{C_Y}{h}$.  We write
$\skel(P,I_X,\prec,N_X,U_X)$ as $\skel_X(P,I,\prec,N,U)$ in what
follows.

\paragraph{Homomorphisms.}
Let $k_0=\skel_X(P,I_0,\prec_0,N_0,U_0)$ and
$k_1=\skel_Y(P,I_1,\prec_1,\break N_1,U_1)$ be skeletons, and
let~$\Theta_0$ and~$\Theta_1$ be the strand spaces associated
with~$I_0$ and~$I_1$.  There is a \index{homomorphism}\emph{skeleton
  homomorphism} $(\phi,\sigma)\typ k_0\mapsto k_1$ if~$\phi$
and~$\sigma$ are maps with the following properties:
\begin{enumerate}
\item\label{item:strand} $\phi$ maps strands of~$k_0$ into those
  of~$k_1$, and nodes as $\phi((s,i))=(\phi(s),i)$, that is $\phi$ is
  in $\sdom(\Theta_0)\to\sdom(\Theta_1)$;
\item\label{item:msg} $\sigma\in\alg{A}_X\to\alg{A}_Y$ is a message
  algebra homomorphism;
\item\label{item:node} $n\in\fn{nodes}(\Theta_0)$ implies
  $\sigma(\evt_{\Theta_0}(n))=\evt_{\Theta_1}(\phi(n))$;
\item\label{item:order} $n_0\prec_0
n_1$ implies $\phi(n_0)\prec_1\phi(n_1)$;
\item\label{item:non} $\sigma(N_0)\subseteq N_1$;
\item\label{item:uniq} $t\in U_0$ implies $\sigma(t)\in U_1$ and
  $\phi(\orig_{k_0}(t))=\orig_{k_1}(\sigma(t))$;
\end{enumerate}
where $\orig_k(t)$ is the node of the event at which~$t$ originates.
Property~\ref{item:uniq} says the node at which an atom uniquely
originates is preserved by homomorphisms.

The definition of a skeleton homomorphism can be extended so that a
bundle can be in the range.  In this case, the range of the message
algebra homomorphism is the initial algebra~\alg{A}.
Property~\ref{item:non} and~\ref{item:uniq} require small tweaks: for
non-origination, $t\in N_0$ implies $\fn{non}(\Theta_1,\sigma(t))$,
and for unique origination, $t\in U_0$ implies
$\fn{uniq}(\Theta_1,\sigma(t),\phi(\orig_{k_0}(t)))$.  Notice that a
homomorphism between skeletons preserves the protocol.  For the case
of a bundle in the range, we require that it be a run of the protocol
of the skeleton.  Let~$\pt(k)$ be~$P$, the protocol of~$k$, so that
the final condition can be written as $\Upsilon\in\run_{\pt(k)}$.  The
bundles associated with skeleton~$k$ are $\{\Upsilon\mid
\some{\delta}\delta\typ k\mapsto\Upsilon \}$.

When given a point-of-view skeleton~$k_0$, if {\cpsa} terminates, it
produces a shape analysis of the form $\delta_i\typ k_0\mapsto k_i$.
The skeletons~$k_i$ are the shapes of this protocol analysis, and they
specify all of the non-adversarial behavior associated with a run
compatible with the point-of-view skeleton.  The shape analysis is
\index{complete shape analysis}\emph{complete} if for all~$\Upsilon$
and $\delta$, $\delta\typ k_0\mapsto\Upsilon$ iff
$\some{i,\delta'}\delta'\typ k_i\mapsto\Upsilon$.
See~\cite{cpsatheory11} for a proof of {\cpsa}'s completeness.

\paragraph{Shape Analysis Sentences.}
The results of a shape analysis are imported into {\pvs} by translating
the analysis into a sentence that is asserted as an axiom in {\pvs},
justified by the fact that the shape analysis is complete.  The
translation is similar to the one appearing in~\cite{Ramsdell12},
however this one is superior due to the foundation provided by the
bundle-based strand space theory presented earlier.  Much of the
translation is simply valid by definition.  Pay particular attention
to the translation of instances.

We define~$\form_\Upsilon(k)=(Y,\Phi)$, where~$\Phi$ is $k$'s
\index{skeleton formula}skeleton formula, and~$Y$ is the formula's set
of variables along with their sorts.  Let $k=\skel_X(P,I,\prec,\break N,U)$.
The set $Y$ is~$X$ augmented with a fresh variable~$z_s$ for each
strand $s\in\sdom(I)$.  In formulas, $z_s$ ranges over
$\sdom(\Theta)$, where~$\Theta$ is the strand space of~$\Upsilon$.
The formula~$\Phi$ is a conjunction of atomic formulas composed as
follows.

\begin{itemize}
\item For each $s\in\sdom(I)$, assert $\fn{htin}(\Theta, z_s, h,
  \sigma(r))$, where $I(s)=\insta(r,h,\sigma)$,
  and\index{htin@\fn{htin}}\index{inst@\fn{inst}} $\fn{htin}(\Theta, z_s,
  h,r)=h\leq|\Theta(s)|\wedge\fn{inst}(\Theta, z_s,r)$.
\item For each $(s,i)\prec(s',i')$, assert
  $(z_s,i)\prec_\Upsilon(z_{s'},i')$.
\item For each $t\in N$, assert $\fn{non}(\Theta,t)$.
\item For each $t\in U$, assert $\fn{uniq}(\Theta, t,(z_s, i))$,
  where $(s,i)=\orig_k(t)$.
\end{itemize}

When $\form_\Upsilon(k)=(X,\Phi)$, the
predicate~$\Sigma_k=\funct{\Upsilon}\Upsilon\in\run_{\pt(k)}\land\some{X}\Phi$
is closed.  (In what follows,~$X$ will refer to the set of algebra
variables augmented with strand variables.)  The bundle~$\Upsilon$ is
a pair~$(\Theta,\to)$, so the strand space~$\Theta$ is the first
element of the pair, and~$\prec_\Upsilon$ is derived from the
communication edges~$\to$ and the strand succession edges in~$\Theta$.

The formula describing a skeleton is order-sorted.  A truth assignment
that tells one how to interpret each skeleton formula must account for
this fact.  As such, the domain of discourse for
\index{interpretation}interpretation~$\interp(\Upsilon)$ contains the
\index{carrier set}carrier set for each sort in the initial message
algebra.  Additionally, for~$\Upsilon=(\Theta,\to)$, the domain of
discourse includes the set~$\sdom(\Theta)$, used to interpret strand
variables~$z_s$.  The interpretation of predicates and function
symbols follows the case of a many-sorted
algebra~\cite[Section~4.3]{Enderton01}.
See~\cite[Section~4]{GoguenMeseguer92} for a description of the
reduction of an order-sorted algebra to a many-sorted algebra.

\begin{thm}\label{thm:skeleton formula}\index{Skeleton Formula Theorem}
Let $\form_\Upsilon(k)=(X,\Phi)$ and
$\Sigma_k=\funct{\Upsilon}\Upsilon\in\run_{\pt(k)}\land\some{X}\Phi$.
For all bundles~$\Upsilon$, $\Sigma_k(\Upsilon)$ iff there is a
homomorphism from~$k$ to $\Upsilon$, i.e.\
$$\Sigma_k(\Upsilon)\Longleftrightarrow
\some{\delta}\delta\typ k\mapsto\Upsilon.$$
\end{thm}

Thus $\{\Upsilon\mid\Sigma_k(\Upsilon)\}$ is another way to
specify the bundles associated with skeleton~$k$.

The intuition behind this proof is the observation that there is an
intimate relationship between the homomorphism and the \index{variable
  assignment}variable assignment used to interpret existentially
quantified variables.

\begin{proof}
Consider the backward implication first.  We are
given~$k=\skel(P,I,\prec,\break N,U)$, $\delta=(\phi,\sigma)$,
and~$\Upsilon=(\Theta,\to)$ such that $\delta\typ
k\mapsto\Upsilon$. To interpret formula~$\Phi$, construct the variable
assignment~$\alpha$ as follows.  For each strand variable~$z_s$,
$\alpha(z_s)=\phi(s)$.  Each algebra variable~$x$ has a corresponding
logical variable, so $\alpha(x)=\sigma(x)$.

The interpretation~$\interp(\Upsilon)$ satisfies~$\Phi$ with~$\alpha$
if each conjunct does so.  For some $s\in\sdom(I)$, consider the
atomic formula $\fn{htin}(\Theta, z_s, h, \sigma'(r))$, where
$I(s)=\insta(r,h,\sigma')$.  Its interpretation is $\fn{htin}(\Theta,
\alpha(z_s), h, \alpha(\sigma'(r)))$ which is
$\fn{htin}(\Theta, \phi(s), h, \sigma(\sigma'(r)))$.  By
definition, $\fn{htin}(\Theta, \phi(s), h, \sigma
(\sigma'(r)))=h\leq|\Theta(\phi(s))|\wedge\fn{inst}(\Theta, \phi(s),
\sigma(\sigma'(r)))$.  By Property~\ref{item:strand} in the
definition of a homomorphism, the length of strand~$\phi(s)$ must be
greater than or equal to~$h$.  Let $r=\role(C_Y,U_Y,N_Y)$.  Recall
that $\fn{inst}(\Theta, \phi(s),\sigma(\sigma'(r)))$ implies
that $\sigma\circ\sigma'\circ\prefix{C_y}{h'}=\Theta(\phi(s))$, where
$h'=|\Theta(\phi(s))|$, which is true by Property~\ref{item:node}.

For the~$\prec_\Upsilon$ predicate, Property~\ref{item:order} in the
definition of homomorphism applies, for \fn{non}, it's the tweak of
Property~\ref{item:non}, and for \fn{uniq}, it's the tweak of
Property~\ref{item:uniq}.

Now consider the forward implication in the theorem.  In this case, we
are given the variable assignment~$\alpha$ such
that~$\interp(\Upsilon)$ satisfies~$\Phi$ with~$\alpha$ and must
construct the corresponding homomorphism.  For each strand
variable~$z_s$, $\phi(s)=\alpha(z_s)$.  Each algebra variable~$x$ has a
corresponding logical variable, so $\sigma(x)=\alpha(x)$.

With this definition of~$\delta$, we show that
$\delta\typ k\mapsto\Upsilon$.
Substitution~$\sigma$ is a message algebra homomorphism, thus
demonstrating Property~\ref{item:msg}.

\begin{sloppypar}
For all $s\in\sdom(I)$, assume~$\interp(\Upsilon)$ satisfies
$\fn{htin}(\Theta, z_s, h, \sigma'(r))$ with~$\alpha$, where
$I(s)=\insta(r,h,\sigma')$.  Therefore, $\fn{htin}(\Theta,
\alpha(z_s), h, \alpha(\sigma'(r)))$ is true, and so is
$\fn{htin}(\Theta, \phi(s), h, \sigma(\sigma'(r)))$ and by
definition $h\leq|\Theta(\phi(s))|$ and $\fn{inst}(\Theta,
\phi(s),\sigma(\sigma'(r)))$.  The height restriction
$h\leq|\Theta(\phi(s))|$ ensures~$\phi$ maps correctly as prescribed
in Property~\ref{item:strand}.  Consider node $n=(s,i)$ in~$k$.  The event
in~$k$ at~$n$ is~$\sigma'(C_Y(i))$ where $r=\role(C_Y,U_Y,N_Y)$.  The
\fn{inst} assertion implies that event $\evt_{\Theta}(\phi(n))$ is
$\sigma(\sigma'(C_Y(i)))$, thus demonstrating Property~\ref{item:node}.
\end{sloppypar}

Property~\ref{item:order}, \ref{item:non}, and~\ref{item:uniq} are
straightforward.
\end{proof}

In what follows, a sentence that universally quantifies a bundle, as
in $\all{\Upsilon}\Phi$, is true if for all~$\Upsilon$,
$\interp(\Upsilon)$ models~$\Phi$.\index{true sentence} Define
$\models_{\interp(\Upsilon)}\Phi$ to mean $\interp(\Upsilon)$
models~$\Phi$, and $\models_{\interp(\Upsilon)}\Phi$ with~$\alpha$ to
mean $\interp(\Upsilon)$ satisfies~$\Phi$ with variable
assignment~$\alpha$.

Given a set of homomorphisms $\delta_i\typ k_0\mapsto k_i$, its shape
analysis sentence $\sent(\delta_i\typ k_0\mapsto k_i)$ is\index{shape analysis sentence}
\begin{equation}\label{eqn:shape sentence}
\all{\Upsilon\in\run_{\pt(k_0)},X_0}\Phi_0 \Longleftrightarrow
\bigvee_i\some{X_i}\Delta_i\wedge\Phi_i,
\end{equation}
where $\form_\Upsilon(k_0)=(X_0,\Phi_0)$.  The same procedure
produces~$X_i$ and~$\Phi_i$ for shape~$k_i$ with one proviso---the
variables in $X_i$ that also occur in~$X_0$ must be renamed to avoid
trouble while encoding the \index{structure preserving map}structure
preserving maps~$\delta_i$.

The structure preserving maps~$\delta_i=(\phi_i,\sigma_i)$ are encoded
in~$\Delta_i$ by a conjunction of equalities.  Map~$\sigma_i$ is coded
as equalities between a message algebra variable in the domain
of~$\sigma_i$ and the term it maps to.  Map~$\phi_i$ is coded as
equalities between strand variables in~$\Phi_0$ and strand variables
in~$\Phi_i$.  Let~$Z_0$ be the sequence of strand variables freshly
generated for~$k_0$, and~$Z_i$ be the ones generated for~$k_i$.  The
strand mapping part of~$\Delta_i$ is
$\bigwedge_{j\in\sdom(\Theta_0)}Z_0(j)=Z_i(\phi_i(j))$,
where~$\Theta_0$ is the strand space associated with~$k_0$.

An example shape analysis sentence is displayed in
Figure~\ref{fig:blanchet's shape analysis sentence}.

\begin{thm}\label{thm:sas}\index{Shape Analysis Sentence Theorem}
If $\delta_i\typ k_0\mapsto k_i$ is a complete shape analysis then
$\sent(\delta_i\typ k_0\mapsto k_i)$ is true.
\end{thm}

\begin{proof}
We show for all bundles~$\Upsilon\in\run_{\pt(k_0)}$,
$\models_{\interp(\Upsilon)}
\all{X_0}\Phi_0\Longleftrightarrow\bigvee_i\some{X_i}\Delta_i\wedge\Phi_i$,
which reduces to showing $\models_{\interp(\Upsilon)}
\Phi_0\Longleftrightarrow\bigvee_i\some{X_i}\Delta_i\wedge\Phi_i$
with~$\alpha$ for all variable assignments~$\alpha$ for $X_0$.  Take
cases on the truth of $\models_{\interp(\Upsilon)}\Phi_0$
with~$\alpha$.

\begin{figure}
$$\xymatrix{k_0\ar[r]^{\delta_i}\ar[dr]_{\delta'_0}&k_i\ar[d]^{\delta'_i} \\
&\Upsilon}$$
\caption{Homomorphism Diagram}\label{fig:homomorphism}
\end{figure}
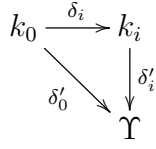

When true, by the proof of Theorem~\ref{thm:skeleton formula},
$\alpha$ specifies the homomorphism $\delta'_0\typ
k_0\mapsto\Upsilon$.  Because the shape analysis is complete, for
some~$i$, $\delta'_i\typ k_i\mapsto\Upsilon$.  By
Theorem~\ref{thm:skeleton formula}, $\models_{\interp(\Upsilon)}
\some{X_i}\Phi_i$ and therefore $\models_{\interp(\Upsilon)} \Phi_i$
with~$\alpha_i$, where~$\alpha_i$ is the variable assignment derived
from~$\delta'_i$.  Let $\alpha\oplus\alpha_i$ be the union of the
mappings in~$\alpha$ and~$\alpha_i$ (the domains of $\alpha$
and~$\alpha_i$ are disjoint).  The proof of this case is complete when
we show $\models_{\interp(\Upsilon)} \Delta_i$
with~$\alpha\oplus\alpha_i$.  Recall that $\delta_i\typ k_0\mapsto
k_i$ and let~$\delta_i=(\phi_i,\sigma_i)$.  See
Figure~\ref{fig:homomorphism} and note that
$\delta'_0=\delta'_i\circ\delta_i$.  For each variable~$x$ in the
domain of~$\sigma_i$, $\Delta_i$ contains the equation
$x=\sigma_i(x)$.  Its interpretation is
$\alpha(x)=\alpha_i(\sigma_i(x))$.  In other words,
$\sigma'_0(x)=\sigma'_i(\sigma_i(x))$, because
$\sigma'_0=\sigma'_i\circ\sigma_i$.  For each strand~$j$ in~$k_0$,
$\Delta_i$ contains the equation $Z_0(j)=Z_i(\phi_i(j))$.  Its
interpretation is $\alpha(Z_0(j))=\alpha_i(Z_i(\phi_i(j)))$.  In other
words, $\phi'_0(j)=\phi'_i(\phi_i(j))$, because
$\phi'_0=\phi'_i\circ\phi_i$.

When $\not\models_{\interp(\Upsilon)}\Phi_0$ with~$\alpha$, there is
no homomorphism of the form $\delta'_0\typ k_0\mapsto\Upsilon$.
Suppose for some~$i$, $\delta'_i\typ k_i\mapsto\Upsilon$.  Then
$\delta'_0=\delta'_i\circ\delta_i$ is a contradiction, so, for
all~$i$, $\delta'_i\typ k_i\not\mapsto\Upsilon$.  By
Theorem~\ref{thm:skeleton formula}, $\not\models_{\interp(\Upsilon)}
\some{X_i}\Phi_i$ and therefore $\not\models_{\interp(\Upsilon)}
\some{X_i}\Phi_i$ with~$\alpha$ implying there is no disjunct on the
R.H.S. that is true.
\end{proof}

\paragraph{Listeners.}

The relationship between skeletons and bundles is not as tidy as
previously described.  {\cpsa} supports something called
\index{listener strands}listener strands that do not appear in
bundles.  A listener strand in a skeleton is an artificial strand used
to assert that some message is available to the adversary.  A listener
strand has length two, and the second event is the transmission of the
message received by the first event.

When translating a listener strand into a bundle, one simply asserts
the existence of a node in the bundle that transmits the strand's
message, and that node inherits the node orderings associated with
the nodes of the listener strand.

The definition of a homomorphism into a bundle requires adjustment to
allow for the disappearance of listener strands.  In particular, the
definition of a homomorphism must use the roles in instances to
identify listener strands.

\section{Detailed Example}\label{sec:example}

The simple example protocol is now revisited.
$$\begin{array}{r@{{}:{}}l}
A\to B&\enc{\enc{s}{\invk{a}}}{b}\\
B\to A&\enc{d}{s}
\end{array}$$
Symmetric key~$s$ is freshly generated, asymmetric keys $\invk{a}$ and
$\invk{b}$ are uncompromised, and the goal of the protocol is to keep
data~$d$ secret.  The {\pvs} description of the protocol in
Eq.~\ref{eq:protocol}, has an initiator and a responder role.  The
role items are:

\begin{equation}\label{eq:cpsa protocol}
\begin{array}{r@{{}={}}l}
\init(a,b\typ\srt{A},s\typ\srt{S},d\typ\srt{D})&
\role(\langle\outbnd\enc{\enc{s}{\invk{a}}}{b},\inbnd\enc{d}{s}\rangle,
\langle\emptyset,\emptyset\rangle,
\langle\{s\},\emptyset\rangle)\\
\resp(a,b\typ\srt{A},s\typ\srt{S},d\typ\srt{D})&
\role(\langle\inbnd\enc{\enc{s}{\invk{a}}}{b},\outbnd\enc{d}{s}\rangle,
\langle\emptyset,\emptyset\rangle,
\langle\emptyset,\emptyset\rangle)
\end{array}
\end{equation}

The {\init} role is
$\{r\mid\some{a,b\typ\srt{A},s\typ\srt{S},d\typ\srt{D}}
r=\init(a,b,s,d)\}$ and the {\resp} role is analogous.  This rendition
of each role ensures it is \index{template inspired role}template inspired.

In this protocol, the unique origination assumption is specified in
the {\init} role, while the two non-origination assumptions are
specified in skeletons.

\begin{figure}
$$\begin{array}{@{}r@{}c@{}l@{}}
k_0&{}={}&\skel_X(\begin{array}[t]{@{}ll}
\{\init(a_0,b_0,s_0,d_0),\resp(a_1,b_1,s_1,d_1)\},
&\mbox{Protocol}\\
\seq{\insta(\resp,2,\{a_1\mapsto a,b_1\mapsto b,s_1\mapsto s,d_1\mapsto d\})},
&\mbox{Instances}\\
\emptyset,
&\mbox{Node orderings}\\
\{\invk{a},\invk{b}\},
&\mbox{Non-origination}\\
\emptyset)
&\mbox{Unique origination}
\end{array}\\
&&\mbox{where $X=a,b\typ\srt{A},s\typ\srt{S}, d\typ\srt{D}$}\\
k_1&{}={}&\skel_Y(\begin{array}[t]{@{}ll}
\{\init(a_0,b_0,s_0,d_0),\resp(a_1,b_1,s_1,d_1)\},
&\mbox{Protocol}\\
\langle\begin{array}[t]{@{}l}
\insta(\resp,2,\{a_1\mapsto a,b_1\mapsto b,s_1\mapsto s,d_1\mapsto d\}),\\
\insta(\init,1,\{a_0\mapsto a,b_0\mapsto b',s_0\mapsto s,d_0\mapsto d'\})\rangle
\end{array}
&\begin{array}[t]{@{}l}
\mbox{Instances}\\
\mbox{\emph{Note $b_0$ is $b'$ not $b$!}}
\end{array}\\
\{(1,0)\prec(0,0)\},
&\mbox{Node orderings}\\
\{\invk{a},\invk{b}\},
&\mbox{Non-origination}\\
\{s\})
&\mbox{Unique origination}
\end{array}\\
&&\mbox{where $Y=a,b,b'\typ\srt{A},s\typ\srt{S}, d,d'\typ\srt{D}$}\\
\delta_1&{}={}&(\seq{0},\{a\mapsto a, b\mapsto b, s\mapsto s, d\mapsto d\})
\end{array}$$
\caption{Shape Analysis for the Simple Example Protocol}\label{fig:blanchet's shape analysis}
\end{figure}

The protocol was constructed with a known flaw for expository
purposes, and as a result the secret is exposed due to an
authentication failure.
The desired authentication goal is:
$$\begin{array}{l}
\all{(\Theta,\to)\in\run_{\pt(k_0)}, a,b\typ\srt{A}, s\typ\srt{S},
  d\typ\srt{D}, z\in\sdom(\Theta)}\\
\quad\fn{htin}(\Theta,z,2,resp(a,b,s,d))\wedge
\cn{non}(\Theta,\invk{a})\wedge\cn{non}(\Theta,\invk{b})\\
\qquad\supset\some{a_0\typ\srt{A}, s_0\typ\srt{S}, d_0\typ\srt{D},
  z_0\in\sdom(\Theta)}\fn{htin}(\Theta,z_0,1,init(a_0,b,s_0,d_0))
\end{array}$$
that is, when the responder~($B$) runs to completion, there is an
initiator~($A$) that is using~$b$ for the encryption of its initial
message.

\begin{figure}
$$\begin{array}{l}
\all{(\Theta,\to)\in\run_{\pt(k_0)}, a_0,b_0\typ
\srt{A}, s_0\typ\srt{S}, d_0\typ\srt{D}, z_0\in\sdom(\Theta)}\\
\quad\fn{htin}(\Theta,z_0,2,\resp(a_0,b_0,s_0,d_0))\wedge
\fn{non}(\Theta,\invk{a_0})\wedge\fn{non}(\Theta,\invk{b_0})\\
\quad\Longleftrightarrow\\
\quad\some{a_1,b_1,b_2\typ\srt{A}, s_1\typ\srt{S},
  d_1,d_2\typ\srt{D}, z_1,z_2\in\sdom(\Theta)}\\
\qquad z_0=z_1\wedge a_0=a_1\wedge b_0=b_1\wedge s_0=s_1\wedge d_0=d_1\wedge{}\\
\qquad\fn{htin}(\Theta,z_1,2,\resp(a_1,b_1,s_1,d_1))\wedge{}\\
\qquad\fn{htin}(\Theta,z_2,1,\init(a_1,b_2,s_1,d_2))\wedge{}\\
\qquad(z_2,0)\prec_{(\Theta,\to)}(z_1,0)\wedge\fn{uniq}(\Theta,s_1,(z_2,0))\wedge{}\\
\qquad\fn{non}(\Theta,\invk{a_1})\wedge
\fn{non}(\Theta,\invk{b_1})
\end{array}$$
\caption{Shape Analysis Sentence for the Simple Example Protocol}\label{fig:blanchet's shape analysis sentence}
\end{figure}

To investigate this goal, we ask {\cpsa} to find out what other
regular behaviors must occur when a responder runs to completion by
giving {\cpsa} skeleton~$k_0$ in Figure~\ref{fig:blanchet's shape
  analysis}.  {\cpsa} produces shape~$k_1$ that shows that an
initiator must run, but it need not use the same key to encrypt its
first message.  The shape analysis sentence for this scenario is
displayed in Figure~\ref{fig:blanchet's shape analysis sentence}.
Needless to say, the authentication goal cannot be deduced from this
sentence due to the man-in-the-middle attack discussed earlier.
However, one can prove the security goal is false by
constructing a bundle that contains the man-in-the-middle attack
specified with the help of adversarial stands, and using it as a
counterexample to the security goal.  If one repeats the
analysis using the protocol in Eq.~\ref{eq:amended protocol}, the
generated shape analysis sentence can be used to deduce the
authentication goal.

\section{Discussion}\label{sec:discussion}

Theorems~\ref{thm:skeleton formula} and~\ref{thm:sas} correspond to
theorems with the same numbers in~\cite{Ramsdell12}.  There are
several key differences between the two works.  Higher-order logic is
used for shape analysis sentences here, but~\cite{Ramsdell12} follows
the first-order logic, model theoretic approach set out
in~\cite{Guttman09}.  A first-order formulation of this version of
shape analysis sentences is straightforward, but would obscure their
use in {\pvs}.

The second difference is this work uses bundles over initial algebras
for models, whereas the previous works use skeletons over free
algebras.  The shallow embedding of strand space theory in {\pvs}
motivates this choice.

Finally, this work faithfully captures the semantics of the roles of
the protocol being analyzed via the height-instance predicate
\fn{htin}, which is defined using roles as sets of role items.  In
previous works, a role origination assumption was ignored.

\section{Conclusion}\label{sec:conclusion}

This paper presented a method for importing security goals and the
results of a {\cpsa} analysis into {\pvs} such that proofs about the
goals can rely on a detailed theory of strand spaces.  The method uses
a shallow embedding of the theory within {\pvs}.  To enable the
embedding, the concept of roles as sets of role items was introduced.
As a result, there is no need to explicitly represent substitutions,
homomorphisms, and skeletons within {\pvs} to prove security goals.
Instead, shape analysis sentences perform the task of transporting
results from {\cpsa} into~{\pvs}.

\section*{Acknowledgment}

Paul D.\ Rowe and Joshua D.\ Guttman provided valuable feedback on
this paper.  I thank Ed Zieglar for his support.

\bibliography{sasbundles}

\begin{thebibliography}{10}

\bibitem{ables2010escrowed}
K.~Ables and M.~Ryan.
\newblock Escrowed data and the digital envelope.
\newblock {\em Trust and Trustworthy Computing}, pages 246--256, 2010.

\bibitem{BlanchetHDR}
Bruno Blanchet.
\newblock {\em V\'{e}rification automatique de protocoles cryptographiques:
  mod\`{e}le formel et mod\`{e}le calculatoire. Automatic verification of
  security protocols: formal model and computational model}.
\newblock M\'{e}moire d'habilitation \`{a} diriger des recherches,
  Universit\'{e} Paris-Dauphine, November 2008.
\newblock En fran\c{c}ais avec publications en anglais en annexe. In French
  with publications in English in appendix.

\bibitem{cremers06}
Casimier~J. Cremers.
\newblock Scyther---semantics and verification of security protocols.
\newblock {Ph.D. Thesis}, Eindhoven Univesity of Technology, 2006.
\newblock \url{http://people.inf.ethz.ch/cremersc/scyther/}.

\bibitem{datta05}
Anupam Datta, Ante Derek, John~C. Mitchell, and Dusko Pavlovic.
\newblock A derivation system and compositional logic for security protocols.
\newblock {\em J. Comput. Secur.}, 13(3):423--482, 2005.

\bibitem{Enderton01}
Herbert~B. Enderton.
\newblock {\em A Mathematical Introduction to Logic}.
\newblock Harcourt/Acedemic Press, second edition, 2001.

\bibitem{GoguenMeseguer92}
Joseph~A. Goguen and Jose Meseguer.
\newblock Order-sorted algebra {I}: Equational deduction for multiple
  inheritance, overloading, exceptions and partial operations.
\newblock {\em Theoretical Computer Science}, 105(2):217--273, 1992.

\bibitem{Guttman09}
Joshua~D. Guttman.
\newblock Security theorems via model theory.
\newblock In {\em Express: Expressiveness in Concurrency, Workshop affiliated
  with Concur}, September 2009.
\newblock Post-proceedings in EPTCS, \url{http://www.eptcs.org/}.

\bibitem{Guttman12}
Joshua~D. Guttman.
\newblock State and progress in strand spaces: Proving fair exchange.
\newblock {\em J. Autom. Reason.}, 48(2):159--195, February 2012.

\bibitem{GuttmanEtAl04}
Joshua~D. Guttman, F.~Javier Thayer, Jay~A. Carlson, Jonathan~C. Herzog,
  John~D. Ramsdell, and Brian~T. Sniffen.
\newblock Trust management in strand spaces: A rely-guarantee method.
\newblock In {\em In Proc. of the European Symposium on Programming (ESOP `04),
  LNCS}, pages 325--339. Springer-Verlag, 2004.

\bibitem{cpsatheory11}
Moses~D. Liskov, Paul~D. Rowe, and F.~Javier Thayer.
\newblock Completeness of {CPSA}.
\newblock Technical Report MTR110479, The MITRE Corporation, March 2011.
\newblock
  \url{http://www.mitre.org/publications/technical-papers/completeness-of-cpsa}.

\bibitem{prover9}
Bill~W. McCune.
\newblock Prover9.
\newblock \url{http://www.cs.unm.edu/~mccune/mace4/}.

\bibitem{cade92-pvs}
{S.} Owre, {J.}~{M.} Rushby, , and {N.} Shankar.
\newblock {PVS:} {A} prototype verification system.
\newblock In Deepak Kapur, editor, {\em 11th International Conference on
  Automated Deduction (CADE)}, volume 607 of {\em Lecture Notes in Artificial
  Intelligence}, pages 748--752, Saratoga, {NY}, jun 1992. Springer-Verlag.
\newblock \url{http://pvs.csl.sri.com}.

\bibitem{Ramsdell12}
John~D. Ramsdell.
\newblock Deducing security goals from shape analysis sentences.
\newblock \url{http://arxiv.org/abs/1204.0480}, April 2012.

\bibitem{cpsa09}
John~D. Ramsdell and Joshua~D. Guttman.
\newblock {CPSA}: A cryptographic protocol shapes analyzer.
\newblock In {\em Hackage}. The MITRE Corporation, 2009.
\newblock \url{http://hackage.haskell.org/package/cpsa}.

\bibitem{ThayerEtal99}
F.~Javier Thayer, Jonathan~C. Herzog, and Joshua~D. Guttman.
\newblock Strand spaces: Proving security protocols correct.
\newblock {\em Journal of Computer Security}, 7(1), 1999.

\end{thebibliography}
\bibliographystyle{plain}

\appendix

\section{Role Annotations}\label{sec:role annotations}

There is a simple extension to the strand space theory in
Section~\ref{sec:strand spaces} that allows the ability to annotate an
event in a role with an object of any type~$\type$.  In practice, few
events in a role need annotation, so for type~\type, events are
associated with the type $\fn{lift}(\type)$.  A \index{lifted
  type}lifted type has two constructors and one accessor, so
$x\in\fn{lift}(\type)$ implies that $x=\bot$ or $x=\up y$ for some
$y\typ\type$.  If $x=\up y$ then $y=\down x$.

Annotations were added by modifying the definition of a role item to
be of the form $\role(C,N,U,A)$, where $C$, $N$, and $U$ are as
before, $A\in\fn{lift}(\type)^+$, and the length of $A$ is the same as
the length of~$C$.  Let role assignment~$\rl$ demonstrate that
bundle~$\Upsilon$ is a run of some protocol.  Node~$n=(s,i)$
in~$\Upsilon=(\Theta,\to)$ is \index{annotations}\emph{annotated} with
$a\in\type$, written $\fn{anno}(\Upsilon,\rl,n,a)$ if
$$\begin{array}{l}
\some{\role(C,N,U,A)\in\rl(s)}\\
\quad\fn{inst}(\Theta,s,\role(C,N,U,A))\wedge A(i)=\up a
\end{array}$$
The set of annotated nodes is
$$\fn{anode}((\Theta,\to),\rl)=\{n\in\fn{nodes}(\Theta)\mid
\some{a\typ\type}\fn{anno}((\Theta,\to),\rl,n,a)\}$$

The annotations can be used to enrich the specification of security
goals.  For example, annotations can be used to combine \index{trust
  management}trust management theories with cryptographic
protocols~\cite{GuttmanEtAl04}.  In this use case, events are
annotated with formulas from a trust management logic.  A formula on
an outbound event is a guarantee and the sender must show the formula
is true before sending the message.  A formula on an inbound event is
an assumption that can be used by the receiver to deduce future
guarantees.  The bundle-based strand space theory can be used to
ensure that whenever a receiver relies on a formula, another principle
has previously guaranteed it.

Role annotations can also be used to reason about state-based
protocols.  The state in the protocol is modeled as a set of states
and a \index{transition relation}transition relation~$\tran$.  An
infinite sequence of states~$\pth$ is a \index{path}\emph{path} if
$\all{i\in\nat}(\pth(i),\pth(i+1))\in\tran$.  To use role annotations
to reason about state, events in roles are annotated with subsets of
the transition relation, that is $\type=\pow{\tau}$.  The art to
making effective use of a state agnostic protocol analyzer is to
modify the message-passing part of the protocol so that a
representation of state is threaded through an execution via
receive-send pairs of strand succession nodes, where the transmitting
node is annotated with a set of transitions consistent with the
threaded state.

A bundle $\Upsilon$ is \index{compatible
  bundle}\emph{compatible}~\cite[Def.~11]{Guttman12} with a
state-based role assignment~{\rl} if there exists $\ell\in\nat$,
$f\in\fn{anode}(\Upsilon,\rl)\to\{0,1,\ldots,\ell-1\}$, and
$\pi\in\fn{path}$ such that
\begin{enumerate}
\item $f$ is bijective,
\item $\all{n_0,n_1\in\fn{anode}(\Upsilon,\rl)}
n_0\prec n_1\Longleftrightarrow f(n_0)<f(n_1)$, and
\item $\all{n\in\fn{anode}(\Upsilon,\rl),a\in\pow{\tau}}$\\
$\mbox{}\quad\fn{anno}(\Upsilon,\rl,n,a)\supset
(\pi(f(n)),\pi(f(n)+1))\in a$.
\end{enumerate}
This definition ties together the state and message-passing worlds and
allows for the verification of state sensitive security goals.  An
in-depth paper describing this technique by Dan Dougherty, Joshua
Guttman, Paul Rowe, and this author is forthcoming.

This appendix ends with a simple example of a stateful protocol called
the Award Card Protocol ({\acp}) created by Joshua Guttman and the
author.  The state in this protocol is a card with some boxes.  When
the card is issued, no box is checked.  Each time a buyer purchases an
item, the cashier checks one box.  The buyer may redeem the card when
all boxes are checked.  It is assumed that a buyer possesses no more
than one card at any time.

For simplicity, suppose every card has just one box and there are two
interactions with cashiers.  Annotated nodes can be used to prove the
two interactions are totally ordered and there must have been a new
card issued between the cashier interactions.  A sketch of the proof
follows.  The model of state is described first, next the protocol
roles, then the method by which the lemma in the state model is
imported into the strand space world, and finally, the use of a shape
analysis sentence to finish the proof of the security goal.

The model of state is not restricted to a card with one box.
Let~\fn{bx} be the number of boxes on a card.  Each
state~$s\in\upsilon$ is the number of unchecked boxes.  The transition
relation is $\tau=\{(s_0,s_1)\mid s_0=s_1+1\lor s_1=\fn{bx}\}$, that
is one box can be checked, or a new card can be issued when one is
redeemed or lost.  The following lemma can be proved by induction.
\begin{lem}[Check or Issue]\label{lem:check or issue}\index{Check or
    Issue Lemma}
$$\begin{array}{l}
\all{\pi\in\fn{path},i,k\in\nat}\\
\quad i\leq k\supset{}\\
\qquad \pi(i)\geq \pi(k)\lor{}\\
\qquad \some{j\in\nat}
i<j\land j\leq k\land \pi(j)=\fn{bx}
\end{array}$$
\end{lem}
In words, either a card has less checked boxes than a predecessor or
there must have been a new card transition in between.

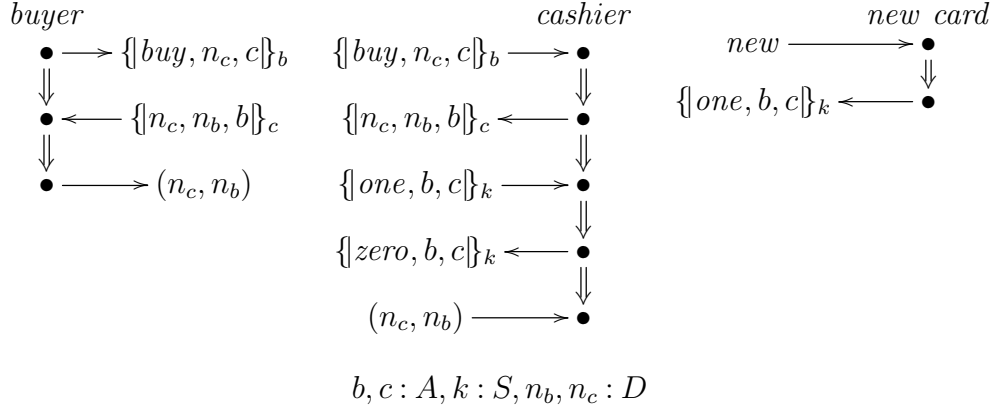
\begin{figure}
$$\begin{array}{rll}
\xymatrix@=.6em{\raisebox{-1ex}[0ex][0ex]{\strut\fn{buyer}}\\
  \bullet\ar@{=>}[d]\ar[r]&\enc{\mathit{buy},n_c,c}{b}\\
  \bullet\ar@{=>}[d]&\enc{n_c,n_b,b}{c}\ar[l]\\
  \bullet\ar[r]&(n_c,n_b)}&
\xymatrix@=.6em{&\raisebox{-1ex}[0ex][0ex]{\strut\fn{cashier}}\\
\enc{\mathit{buy},n_c,c}{b}\ar[r]&\bullet\ar@{=>}[d]\\
\enc{n_c,n_b,b}{c}&\bullet\ar@{=>}[d]\ar[l]\\
\enc{\mathit{one},b,c}{k}\ar[r]&\bullet\ar@{=>}[d]\\
\enc{\mathit{zero},b,c}{k}&\bullet\ar@{=>}[d]\ar[l]\\
(n_c,n_b)\ar[r]&\bullet}&
\xymatrix@=.6em{&\raisebox{-1ex}[0ex][0ex]{\strut\fn{new\ card}}\\
\fn{new}\ar[r]&\bullet\ar@{=>}[d]\\
\enc{\mathit{one},b,c}{k}&\bullet\ar[l]}
\end{array}$$
$$b,c\typ A, k\typ S, n_b,n_c\typ D$$
\caption{Award Card Protocol Traces}\label{fig:acp traces}
\end{figure}

The Award Card Protocol requires an addition to the signature in
Figure~\ref{fig:signature}---an infinite set of constants~$\cn{g}_i$
of sort~$\top$ called \index{tags}tags.  This protocol uses four tags,
$\fn{zero}=\cn{g}_0$, $\fn{one}=\cn{g}_1$, $\fn{buy}=\cn{g}_2$, and
$\fn{new}=\cn{g}_3$.

There are three roles in the {\acp}, a new card issuer, a cashier, and
a buyer.  The trace of each role is displayed in Figure~\ref{fig:acp
  traces}.

An interaction between a cashier and a buyer is authenticated using a
Needham-Schroeder-Lowe inspired message pattern.  Ignore the third and
fourth event in the cashier role to see the pattern.

The remainder of the events in the roles encode the state, most using
the encoding produced by the injective function~$g(s)=\cn{g}_s$.  The
third and fourth event in the cashier role encode a box checking
transition.  The first and second event in the new card role encode a
new card transition, where the first event is a dummy value due to the
special form of a new card transition.

In general, state encoding message events are inbound followed by
outbound event pairs.  The outbound event of the pair is annotated.
If~$i$ is the index of the outbound event of trace~$C$, then it is
annotated with $\{(s_0,s_1)\mid g(s_0)=h(C(i-1))\land
g(s_1)=h(C(i))\}$, where~$h$ extracts the portion of the message from
an event that encodes the state.  In the special case of events of the
form of a new card transition, the outbound event is annotated with
$\{(s_0,s_1)\mid g(s_1)=h(C(i))\}$.


\begin{figure}
\begin{lem}[Bridge]\label{lem:bridge}\index{Bridge Lemma}
$$\begin{array}{l}
\all{\Upsilon,\rl}
\fn{compatible}(\Upsilon,\rl)\supset{}\\
\quad\all{n_0,n_1\in\fn{anode}(\Upsilon,\rl),a_0,a_1\in\pow{\tau},
s_0,s_1\in\upsilon}\\
\qquad\fn{anno}(\Upsilon,\rl,n_0,a_0)\land
\fn{anno}(\Upsilon,\rl,n_1,a_1)\land n_0\prec n_1\land{}\\
\qquad a_0\subseteq\{(s_2,s_3)\mid s_3=s_0\}\land
a_1\subseteq\{(s_2,s_3)\mid s_2=s_1\}\supset{}\\
\qquad\quad s_0\geq s_1\lor{}\\
\qquad\quad\some{n\in\fn{anode}(\Upsilon,\rl)}\\
\qquad\qquad\fn{anno}(\Upsilon,\rl,n,\{(s_2,s_3)\mid
s_3=\fn{bx}\})\land{}\\
\qquad\qquad n_0\prec n\land n\prec n_1
\end{array}$$
\end{lem}
\end{figure}

When analyzing the {\acp}, {\cpsa} has no means by which to enforce
the linear ordering of state encoding nodes in bundles, and it may
produce a shape analysis sentence that is incompatible with our notion
of state.  To verify state aware security goals, we will restrict our
attention to the bundles that are compatible with the role assignment
implied by the role definitions.  Because function~$f$ in the
definition of compatibility is a bijection, annotated nodes in
compatible bundles must be linearly ordered.

The compatible bundle assumption allows one to infer the existence of
nodes that are not revealed by {\cpsa}.  In the case of the {\acp},
this is done by importing the Check or Issue Lemma into the strand space
world by proving the Bridge Lemma (Lemma~\ref{lem:bridge}).
The proof of the Bridge Lemma makes use of every part of the
definition of compatibility.

The implication in the Check or Issue Lemma corresponds to the second
implication in the Bridge Lemma.  The correspondence of the
conclusions of each implication is straightforward, however, the
hypothesis of the Bridge Lemma is much more complicated than the one
in the Check or Issue Lemma.  Yet all it is saying is that the
beginning and ending states over the range of the path are~$s_0$
and~$s_1$, where as in the Check or Issue Lemma, those states are
simply referred to by~$\pi(i)$ and~$\pi(k)$.

Dear reader, at this point I promised to describe the use of a shape
analysis sentence to complete the proof of the security goal.  I
fibbed.  This example is so simple and contrived, there is no need to
run {\cpsa} at all!  The fact that when there are two interactions
with cashiers, there must have been a new card issued between the
cashier interactions follows from the point-of-view skeleton one would
use to analyze this security goal.  In this respect, this is a very
unusual example.

The above procedure for verifying security goals of protocols with
state has been successfully applied to the Envelope
Protocol~\cite{ables2010escrowed}.  In this case, two shape analysis
sentences are required to prove the most interesting security goal.
The {\pvs} proof is detailed and involved, and relies on fundamental
properties of bundles.

For example, it was shown in {\pvs} that if node~$n_0$ is before some
transmission node~$n_2$, then either the nodes are on the same strand
or there is a reception node~$n_1$ before~$n_2$ on the same strand,
such that $n_0$ is before~$n_1$.  The compatibility assumption implies
a total ordering among transmission nodes with annotations.  The above
lemma is used to infer the correct ordering of nodes that receive
state encoding messages.  The lemma is also used in the proof of the
{\acp} security goal.

The proof of the Envelope Protocol security goal will be described in
the forthcoming paper mentioned earlier.

\printindex

\end{document}